\documentclass[11pt]{article}

\usepackage[letterpaper,margin=1in]{geometry}

\usepackage{authblk}

\usepackage{microtype}
\usepackage{url}
\usepackage{amsmath}
\usepackage{amsfonts}
\usepackage{amssymb}
\usepackage{amsthm}
\usepackage{balance}
\usepackage{hyperref}
\usepackage{subfig}
\usepackage{float}
\usepackage{graphicx}

\newcommand{\BO}[1]{{ O}\left(#1\right)}

\newcommand{\BTO}[1]{\tilde{ O}\left(#1\right)}
\newcommand{\SO}[1]{{o}\left(#1\right)}
\newcommand{\BT}[1]{{\Theta}\left(#1\right)}
\newcommand{\BOM}[1]{\Omega\left(#1\right)}

\newtheorem{lemma}{Lemma}
\newtheorem{theorem}{Theorem}
\newtheorem{corollary}{Corollary}

 \newcommand{\xif}{{\bf{\em{if~}}}}
 \newcommand{\xthen}{{\bf{\em{then~}}}}

 \newcommand{\xfor}{{\bf{\em{for~}}}}
 \newcommand{\xto}{{\bf{\em{to~}}}}
 
 \newcommand{\xdo}{{\bf{\em{do~}}}}

 \newcommand{\xand}{{\bf{\em{and~}}}}

 \newcommand{\T}{\hspace{0.5cm}}

\def\func#1{\textrm{\bf{\sc{#1}}}}

\newcommand{\nvs}{\vspace{0cm}}

\newcommand{\hide}[1]{}

\date{}
\author[1]{Rezaul Chowdhury}
\author[2]{Francesco Silvestri}
\author[3]{Flavio Vella}
\affil[1]{Stony Brook University, USA, rezaul@cs.stonybrook.edu}
\affil[2]{University of Padova, Italy, silvestri@dei.unipd.it}
\affil[3]{Free University of Bozen, Italy, flavio.vella@unibz.it}

\begin{document}

\title{A Computational Model for Tensor Core Units}
\maketitle

\begin{abstract}
To respond to the need of efficient training and inference of deep neural networks, a plethora of domain-specific hardware architectures have been introduced, such as Google Tensor Processing Units and NVIDIA Tensor Cores. A common feature of these architectures is a hardware circuit for efficiently computing a dense matrix multiplication of a given small size.

In order to broaden the class of algorithms that exploit these systems, we propose a computational model, named the TCU model, that captures the ability to natively multiply small matrices. We then use the TCU model for designing fast algorithms for several problems, including  matrix operations (dense and sparse multiplication, Gaussian Elimination), graph algorithms (transitive closure, all pairs shortest distances), Discrete Fourier Transform, stencil computations, integer multiplication, and polynomial evaluation. We finally highlight a relation between the TCU model and the external memory model.
\end{abstract}

\section{Introduction}
Deep neural networks are nowadays used in several application domains where big data are available, and have led to breakthroughs, such as reducing word error rates in speech recognition by 30\% over traditional approaches~\cite{Dean16} and cutting the error rate in an image recognition competition  from 26\% to 3.5\%~\cite{He16}.
The huge data set size, although crucial for improving  neural network quality, rises performance issues during the training and inference steps.
To respond to the increasing computational needs, domain-specific hardware accelerators have been introduced by several IT firms,  such as Google Tensor Processing Units~\cite{jouppi2017datacenter}, NVIDIA Tensor Cores~\cite{nvidia-volta}, Intel KNL's AVX extensions~\cite{Intel}, Apple Neural Engine~\cite{Apple}, and ARM's Machine Learning Processor~\cite{ARM} among the others.
These compute units have been specifically developed for some deep learning models, such as  multilayer perceptrons, convolutional neural networks, and recurrent neural networks.

These accelerators significantly vary in hardware architectures; however, they share  circuits to efficiently multiply small and dense matrices of fixed size.
Indeed, matrix multiplication is one of the most important computational primitives in deep learning.
By using the terminology introduced in~\cite{Dakkak19}, we refer to all accelerators supporting hardware dense matrix multiplication as \emph{Tensor Core Units (TCUs)} (or simply tensor units).
By focusing on a specific computational problem, namely matrix multiplication, TCUs exhibit at the same time both high performance and low energy consumption with respect to traditional CPU or GPU approaches~\cite{jouppi2017datacenter}. 

TCUs are becoming the mainstream technology for deep learning, with constantly decreasing economic costs and a tighter integration with the main processing unit.
Although TCUs were developed for domain-specific problems, it would be interesting and profitable to extend their application domain, for instance by targeting problems from linear algebra, data mining or machine learning (other than deep learning).
A similar scenario appeared with the introduction of GPUs: 
introduced in the 2000s for accelerating computer graphics (in primis, video games), GPUs have been used since then for very different computational problems, like bioinformatics~\cite{Nobile16}, data mining~\cite{Bohm2009}, and neural networks~\cite{Shi16}. 
\emph{Will TCUs have the same wide impact of GPUs?}

The goals of this paper are to present a framework for designing and analyzing efficient algorithms for TCUs,
and to expand the class of algorithms that exploit  TCUs.
We first introduce in Section~\ref{sec:model}  a computational model for tensor core units, which we call \emph{$(m,\ell)$-TCU}, that captures the main features of tensor units:
\begin{enumerate}
    \item \textbf{High performance matrix multiplication.} For a given model parameter $m\geq 1$, two matrices of size $\sqrt{m}\times \sqrt{m}$ can be multiplied in $\BO{m}$ time by using the  hardware circuit in tensor units.
    \item \textbf{Latency cost.} The model parameter $\ell$ captures the latency cost for setting up the tensor unit and preparing the input/output operands. 
    \item \textbf{Asymmetric behavior.} Some tensor units can  efficiently process a left matrix with a large number $n$ of rows (i.e., a tall left matrix). 
    Thus, we let the $(m,\ell)$-TCU to natively multiply an $n \times \sqrt{m}$ matrix by a $\sqrt{m}\times \sqrt{m}$ matrix,
 without splitting the left matrix into submatrices of size $\sqrt{m}\times \sqrt{m}$. The multiplication requires $\BO{n\sqrt{m}}$ time and we let $n\geq \sqrt{m}$ to be a user defined value.
\end{enumerate}

In Section~\ref{sec:algos}, we design several algorithms that exploit tensor accelerators and analyze their performance on the $(m,\ell)$-TCU model. 
More specifically, we show how to compute some matrix operations (dense and sparse multiplication, Gaussian Elimination), graph problems (transitive closure, all pairs shortest distances), the Discrete Fourier Transform, a class of stencil computations, integer multiplication, and polynomial evaluation. 
These algorithms give evidence that TCUs can be potentially used for different computational problems, in addition to deep neural networks.
Finally in Section~\ref{sec:emem}, we observe that some lower bounds on the I/O complexity in the external memory model~\cite{Vitter06} translate into lower bounds on the running time in the TCU model.

We observe that, from a theoretical point of view, TCUs will not be needed if it is developed an algorithm for multiplying two $\sqrt{n}\times \sqrt{n}$ matrices in $\BO{n}$ time (i.e., $\omega = 2$). 
However, from a more realistic point of view, it is very unlikely that such algorithm will have experimental performance equivalent to the state of the art, in particular of hardware implementations.
A rigorous approach to TCUs is an important step to fully exploit tensor accelerators and to further improve the performance of algorithms, but also to better understand the generality of matrix multiplication as computational primitive.

\subsection{Previous results} 
\paragraph{Tensor Core Units.} The literature on tensor core units has mainly focused on architectural issues, see e.g.~\cite{jouppi2017datacenter,zhu2018,Reagen17}. 
Some works, like~\cite{Markidis18,Raihan19}, have investigated the programming model and performance of deep neural networks workloads in the NVIDIA Tensor Cores.

To the best of our knowledge, the only papers that broaden the class of algorithms expressible as TCU operations are~\cite{Dakkak19,carrasco19,Sorna18}.
The papers~\cite{Dakkak19,carrasco19} design algorithms for scanning and reduction that exploit NVIDIA tensor cores.
In~\cite{Sorna18}, it is shown how to speed up the Discrete Fourier Transform (DFT)  by exploiting the half precision multiplication capability of NVIDIA tensor cores. The algorithm in~\cite{Sorna18} uses the Cooley-Tukey algorithm where DFTs of size 4 are computed using  tensor cores, and it is a special case of the TCU algorithm proposed in this paper in Section~\ref{sec:FFT}. 
However, none of the previous works has proposed a rigorous method for studying how to accelerate algorithms with TCUs.

\paragraph{Matrix multiplication}
From a practical point of view, the most efficient algorithms for dense matrix multiplication are those based on the definition of matrix multiplication and which require $\BT{n^{3/2}}$ operations for multiplying two $\sqrt{n}\times \sqrt{n}$ matrices (see e.g. the BLAS library).
Nevertheless from a theoretical point of view, several papers have been investigating algorithms requiring $\BO{n^{\omega/2}}$ operations for some $\omega<3$. 
The work of Strassen~\cite{Strassen69} showed that $\omega\leq 2.81$, and then subsequent works have been improving the upper bound on $\omega$, up to  the current best result $\omega<2.3728$~\cite{Williams12,LeGall14}.

Some works, like~\cite{BjorklundPWZ14,Riko15}, investigate how to use fast matrix multiplication to compute  problems like triangle listing  and sparse matrix multiplication.
The results in~\cite{BjorklundPWZ14} show how to list $t$ triangles in a graph with $m$ edges in $\BO{m^{2\omega/(\omega+1)} + m^{3(\omega-1)/(\omega+1)}t^{(3-\omega)/(\omega+1)}}$ time. 
In~\cite{Riko15}, it is  shown how to compute a sparse matrix multiplication in time $\BTO{\sqrt{n}Z^{(\omega-1)/2}+I}$, where $I$ is the number of non-zeros in the input $\sqrt{n}\times \sqrt{n}$ matrices and $Z$ is the number of non-zero in the output matrix.
These algorithms automatically work in the $(m,\ell)$-TCU model by replacing the $\BO{n^{\omega/2}}$ running time of fast matrix multiplication in the RAM model, with the $\BO{n^{\omega/2(m+\ell)}/m^{\omega/2-1}}$ running time in the TCU model (see Theorem~\ref{thm:mmult}).

\section{Preliminaries}\label{sec:prelim}

\subsection{Technical overview on some TCUs}\label{sec:technical}
\iffalse
With the advent of deep learning as the dominant paradigm in machine learning, a plethora of dedicated parallel architectures have been developed for accelerating training and inference phases. 
\fi
We now briefly describe the main characteristics of the most relevant hardware accelerators for deep learning: Google's Tensor Processing Unit (TPU)~\cite{jouppi2017datacenter} and  NVIDIA's Tensor Cores~\cite{nvidia-volta} (TCs). They are used to accelerate convolution layers and the related matrix multiplication operations, which represent the most computationally expensive part of deep learning applications.

The Tensor Processing Unit (TPU) is an application-specific integrated circuit developed for accelerating the inference phase of deep neural networks. A TPU consists of an ALU Matrix Multiply Unit (MMU)  and two on-chip memories, called Unified Buffer and Weight Memory.
TPU has been designed as an accelerator to plug into a traditional server as GPUs do through PCIe I/O bus. 
%On the contrary to GPUs, the processor sends the instruction buffer instructions by avoiding the use of a fetch unit inside the TPU. 
Data are sent from the CPU host memory to the local TPU memories with the goal of offloading all the computation. 
The MMUs  are composed by $256 \times 256$ $8$-bit multiplier-accumulator units, where the $16$-bit products are stored in $32$-bit accumulators. 
A systolic execution of the units reduces the overhead and maximizes the throughput, with $256$-element partial sum per clock cycle. 
Briefly, the typical TPU workflow is summarized as follows:
(1) Read  a variable-sized $n\times 256$ input, with $n<96$K from the CPU host memory into the Unified Buffer memory;
(2) Read a $256 \times 256$ matrix into Weight Memory and then into the MMUs;
(3) Perform a matrix multiplication;
(4) Write the $n\times 256$ output from the TPU (Unified Buffer) to the CPU host memory.
We observe that, using our notation, step 1 reads the left matrix $A$, while step 2 reads matrix $B$.

In the NVIDIA "Volta" architecture, Tensor Cores extend the traditional GPU architectures and the parallel programming interface (CUDA) by providing dedicated units to efficiently perform dense matrix  multiplication. 
The Volta micro-architecture revised NVIDIA Streaming Multiprocessors (SM) design: the SM consists of 
$4$ processing blocks. Each block contains: $2$ Tensor Cores, $8$ Floating Point (FP) units operating at $64$-bit, $16$ FP operating at $32$-bit, $8$ Integer Unit operating at $32$-bit and one Special Function Unit. 
%(see Figure~\ref{fig:nvidia_sm}). 
Concerning the memory hierarchy, the $L1$ cache and the shared memory are located in the same in-chip surface. The $L2$ is also included in the same die and it is accessed among multiple SMs. The High Bandwidth Memory (HBM) can be addressed by a $4096$-bit memory interface. 
Tensor Cores can perform $64$ floating-point Fused-Multiply-Add (FMA) operations in one cycle. 
FMA operates in half-precision ($16$-bit) and optionally stores a $32$-bit output by using a sum accumulator. Such design can deliver a peak of $31.4$ Tera Floating Point Operations Per Second (FLOPS) with half precision on a Tesla V100 ($640$ Tensor Cores spread over $80$ SM). 
From the programming point of view, although TC basically performs one matrix-multiply and accumulate operation on $4 \times 4$ matrices, it is possible to multiply a $16 \times 16$ matrix at programming level with one CUDA warp ($32$ threads). 
This requires primitives for data loading and the synchronization of the result from/to registers through load and store units. 
%On the contrary to TPUs, which also provide specific units for accelerating other neural networks operations (e.g., Activation Unit), TCs can be used for a generic matrix multiplication. 

\subsection{Systolic algorithms for matrix multiplication}
The circuits which implement matrix multiplication in the Google TPU and in the NVIDIA TC adopt a systolic algorithm for matrix multiplication.
A systolic algorithm is an algorithm for a systolic array, that is a network of processing elements (PEs)
that rhythmically compute and pass data through the system~\cite{Leighton91}. 
%We refer to Appendix~\ref{sec:technical} for a technical overview of the Google TPU and NVIDIA TC accelerators.

For the sake of completeness, we formalize the systolic algorithm implemented in the Google TPU~\cite{jouppi2017datacenter}. 
The implementation on NVIDIA  TCs is slightly different but shares the same high level structure, and we refer to~\cite{Leighton91} for a more complete overview of systolic algorithms.
The systolic algorithm is implemented on a 2-dimensional array of $m$ PEs, and we denote with $p_{i,j}$ the PE at row $i$ and column $j$, for each $0\leq i,j <\sqrt{m}$.
Let $A$ and $B$ be the two $\sqrt{m}\times \sqrt{m}$ input matrices, and let $C=A\cdot B$ be the $\sqrt{m}\times \sqrt{m}$ output matrix; we denote with $a_{i,j}, b_{i,j}, c_{i,j}$ the entry in row $i$ and column $j$ of $A$, $B$, $C$ respectively; for notational simplicity, we let $a_{i,j}=0$ if $i,j<0$ or $i,j\geq \sqrt{m}$.
The algorithm works as follows (see also Figure~\ref{fig:systolic}):

\begin{figure}[t]
\centering
\includegraphics[scale=0.35]{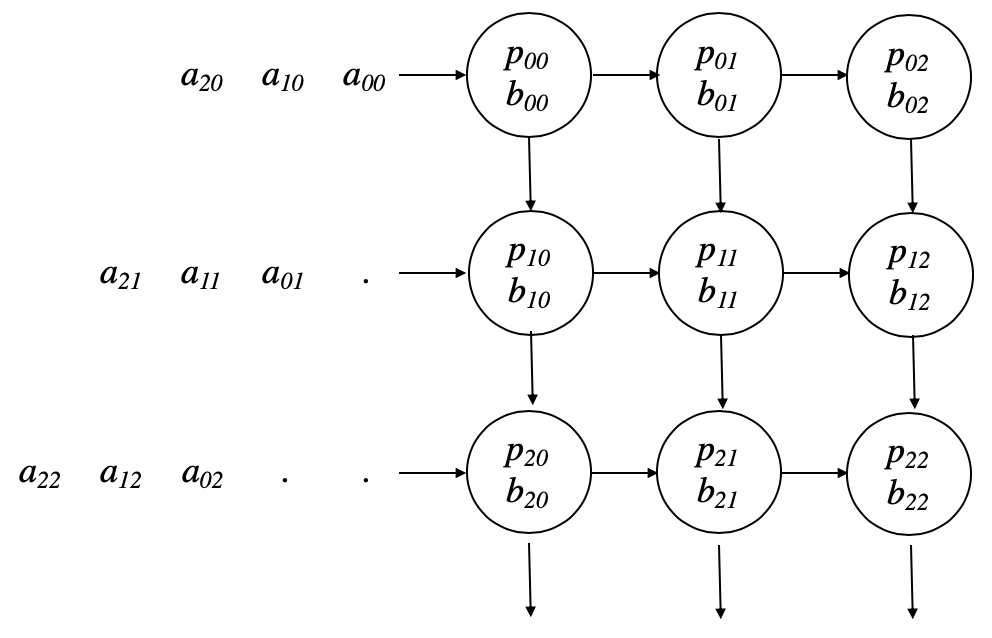}
\caption{A 3x3 systolic array after loading matrix $B$.\label{fig:systolic}} 
\end{figure}

\begin{itemize}
    \item In the first $\sqrt{m}$ steps, matrix $B$ is pushed within the $m$ PEs so that $p_{i,j}$ contains $b_{i,j}$. 
    \item The algorithm then executes $3\sqrt{m}$ steps.
In each step $k$, with $0\leq k<2\sqrt{m}$, each PE $p_{i,j}$ receives: 1) an entry $a$ of $A$ from the left PE $p_{i,j-1}$ or the input $a_{k-i,i}$ if $j=0$; 2) a partial sum $c$ from the top PE $p_{i-1,j}$, or we set $c=0$ if $i=0$.
Then, each  $p_{i,j}$ computes $c = c + a\cdot b_{i,j}$ (recall that $b_{i,j}$ is in the local memory of $p_{i,j}$.
Finally $p_{i,j}$ forwards $a$ to the right PE ($p_{i,j+1}$, if any) and $c$ to the bottom PE ($p_{i+1,j}$ or it is output if $i=\sqrt{n}-1$).
    \item We observe that each  $p_{\sqrt{m}-1, j}$ outputs $c_{i,j}$ at the end of step $k=\sqrt{m}+i+j$. 
\end{itemize}

We observe that the algorithm can be extended to compute $C = A\cdot B$ where $A$ is an $n\times \sqrt{m}$ matrix and $B$ is a $\sqrt{m}\times \sqrt{m}$ matrix, by just continuing pumping all rows within the system.
This feature is not available in the NVIDIA implementation since matrix $B$ does not reside in the local PE memories, but it is percolated within the array as matrix $A$. 

\section[The (m,l)-TCU model]{The $(m,\ell)$-TCU model}\label{sec:model}
We propose a computational model for tensor core units that captures the following three properties.
\paragraph{1: Matrix acceleration.} The hardware circuits  implement a parallel algorithm to multiply two matrices of a fixed size, and the main cost is dominated by reading/writing the input and output matrices. 
    For a given hardware parameter $m$, we have  that the multiplication of two matrices $A$ and $B$ of size $\sqrt{m}\times \sqrt{m}$ are implemented in  time $\BO{m}$. 
    With  time, we mean the running time as seen by the CPU clock and it should not be confused with the total number of operations executed by the unit, which is always $\BT{m^{3/2}}$. Indeed, no existing tensor unit implements fast matrix multiplication algorithms, as for instance Strassen.
    The matrix multiplication operation is called by an instruction specifying the address (in memory) of the two input matrices and of the output matrix where the result will be stored;  data will  be loaded/stored by the tensor unit.
    
    \paragraph{2: Latency cost.} A call to the tensor unit has a latency cost. 
  As the state of the art tensor units use systolic algorithms, the first output entry is computed after $\BOM{\sqrt{m}}$ time. 
    There are also initial costs associated with  activation, which can significantly increase when the unit is not connected to the CPU by the internal system bus or is shared with other CPUs.
    We thus assume that the cost of the multiplication of two matrices  of size $\sqrt{m}\times \sqrt{m}$ is $\BO{m+\ell}$, where $\ell>0$ is the latency cost.

\paragraph{3: Asymmetric behavior.} As tensor units are designed for improving training and inference in deep networks, the two matrices in the multiplication $A \times B$ are managed differently. Matrix $B$ represents the model (i.e., the weights of the deep neural network), while the rows of  matrix $A$ represent the input vectors to be evaluated.  
As the same  model can be applied to $k$ vectors, with $n >> \sqrt{m}$, it is possible to first load the weights in $B$ and then to stream the $n$ rows of $A$ into the tensor unit (possibly in chunks of $\sqrt{m}$ rows), reducing thus latency costs. 
Thus, we assume in our model that two matrices of size $n\times \sqrt{m}$ and $\sqrt{m}\times \sqrt{m}$ are multiplied in time $\BO{n\sqrt{m}+\ell}$, where the number $n$ of rows is defined by the algorithm and $n\geq \sqrt{m}$.\\

More formally, we define the \emph{Tensor Computing Unit (TCU) model} as follows. 
The \emph{$(m,\ell)$-TCU} model is a standard RAM model where the CPU contains a circuit, named tensor unit, for performing a matrix multiplication $A\times B$ of size $n \times \sqrt{m}$ and $\sqrt{m}\times \sqrt{m}$ in time $\BO{n\sqrt{m}+\ell}$, where $m\geq 1$ and $\ell\geq 0$ are two model parameters and $n\geq \sqrt{m}$ is a value (possibly input dependent) specified by the algorithm. 
The matrix operation is initialized by a (constant size) instruction containing the addresses in memory of the two input matrices $A$ and $B$, of the output matrix $C$, and the row number $n$ of $A$.
The \emph{running time} (or simply time) of a TCU algorithm is given by the total cost of all operations performed by the CPU, including all calls to the tensor unit.
We assume no concurrency between tensor unit, memory and CPU, and hence at most one component is active at any time.
Each memory (and TPU) word consists of $\kappa$ bits (in general, we denote $\kappa = \BOM{\log n}$ where $n$ is the input size, that is enough for storing the input size in one word.

\subsection{Discussion on the model}\label{sec:disc}
In this preliminary work, our goal is to understand how to exploit circuits of fixed size for matrix multiplication.
We then do not include in the model some characteristics of existing hardware accelerators, like limited numerical precision and parallel tensor units.
In particular, the modeling of only a single tensor unit can be seen as a major weakness of our model since existing boards contain a large number of TPUs (e.g., more than 500 TC in the Nvidia Titan RTX). 
However, we believe that the first step to exploit tensor accelerators is to investigate which problems can benefit of matrix multiplication circuits; we have then opted for a simple model with only a TCU. 
Moreover, existing hardware accelerators use different parallel architectures and interconnection networks, while they agree on matrix multiplication as main primitive.
We now make some considerations on how Google TPUs and NVIDIA TCs fit our model.

In the Google TPU (in the version described in~\cite{jouppi2017datacenter}), the right matrix $B$ has size $256\times 256$ words (i.e., $m= 65536$). 
The left matrix $A$ is stored in the local unified buffer of $96k\times 256$ words; thus, TPUs can compute the product between two matrices of size $96$k$\times 256$ and  $256\times 256$ in one (tensor) operation.
The number of rows of the left matrix in the TCU model is a user defined parameter (potentially a function of the input size); on the other hand, the number of rows of the left matrix in the TPU is user defined but it is upper bounded by a hardware-dependent value (i.e., 96K). Being this bound quite large, a TPU better exploits a tall left matrix than a short one, as in our TCU model.
The systolic array works in low precision with $8$ bits per word ($\kappa = 8$).
The bandwidth between CPU and TPU was limited in the first version (16GB/s), but it is significantly higher in more recent versions (up to $600$ GB/s). 
Although TPU has a quick response time, the overall latency is high because the right hand matrix has to be suitably encoded via a TensorFlow function before loading it within the TPU: in fact, the TPU programming model is strongly integrated with TensorFlow, and it does not allow to use bare matrices as inputs.
The high latency cost might mitigate the fact that our model does not capture limited bandwidth.

The programming model of the NVIDIA Volta architecture allows one to multiply matrices of size $16\times 16$, although the basic hardware unit works on $4\times 4$ matrices; we thus have $m=256$. 
Memory words are of $\kappa =16$ bits.
TCs exhibit high bandwidth and low latency, as data are provided by a  high bandwidth memory shared with the GPU processing units. 
Matrices $A$ and $B$ can be loaded within TCs without a special encoding as in TPUs, since the NVIDIA Volta natively provides support for matrix multiplication. 
Finally we observe that, as TCs are within a GPU, any algorithm for TCs has also to take into account GPU computational bottlenecks (see e.g.~\cite{KarsinWCIS18,AfshaniS15}).

\section{Algorithms}\label{sec:algos}

\subsection{Matrix multiplication}

\subsubsection*{Dense matrix multiplication}
A Strassen-like algorithm for matrix multiplication is defined in~\cite{Ballard13} as a recursive algorithm that utilizes as base case an algorithm $\mathcal A$ for multiplying two $\sqrt{n_0}\times \sqrt{n_0}$ matrices using $p_0$ element multiplications and $\BO{n_0}$ other operations (i.e., additions and subtractions); we assume $n_0=\BO{p_0}$.
Given two $\sqrt{n}\times \sqrt{n}$ matrices with $n>n_0$, a Strassen-like algorithm envisions the two $\sqrt{n}\times \sqrt{n}$ matrices as two matrices of size  $\sqrt{n_0}\times \sqrt{n_0}$ where each entry is a submatrix of size  $\sqrt{n/n_0}\times \sqrt{n/n_0}$: then, the algorithm recursively computes $p_0$ matrix multiplications on the submatrices (i.e., the $p_0$ element multiplications in $\mathcal A$) and then performs $\BO{n}$ other operations.
For given parameters $p_0$ and $n_0$, the running time of the algorithm is 
$T(n) = \BO{n^{\omega_0}}$,  where\footnote{We observe that $\omega_0$ corresponds to $\omega/2$, where $\omega$ is the traditional symbol used for denoting the exponent in fast matrix multiplication algorithms.} $\omega_0 = \log_{n_0} p_0$.
By setting $n_0=4$ and $p_0 = 8$, we get the standard matrix multiplication algorithm ($\omega_0=3/2$), while with  $n_0=4$ and $p_0=7$ we get the Strassen algorithm ($\omega_0=\log_4 7 \sim 1.403$). 
Any fast matrix multiplication algorithm can be converted into a Strassen-like algorithm~\cite{Raz03}.

The TCU model can be exploited in Strassen-like algorithms by ending the recursion as soon as a subproblem fits the tensor unit: when $n \leq m$, the two input matrices are loaded in the tensor unit and the multiplication is computed in $\BO{m}$ time.
We assume $m\geq n_0$, otherwise the tensor unit would not be used.

\begin{theorem}\label{thm:mmult}
Given a Strassen-like algorithm with parameters $n_0$ and $p_0$, then there exists a TCU algorithm that multiplies two $\sqrt{n}\times \sqrt{n}$ matrices on an $(m,\ell)$-TCU model, with $m\geq n_0$, in time
$$
T(n) = \BO{\left(\frac{n}{m}\right)^{\omega_0}(m+\ell)}.
$$
\end{theorem}
\begin{proof}
The running time is given by the following simple recursion
which assumes $n \geq m$:
$$
T(n) = \left\{ 
\begin{array}{ll}
% p_0 T(n/n_0) + \BO{n} & \textnormal{if } n>m\\
% \BO{m+\ell}     &  \textnormal{if } n\leq m
%
\BO{ \frac{n^{3/2}}{m^{1/2}} + \frac{n}{m} \ell} &  \textnormal{if } m \leq n \leq m n_0;\\
p_0 T(n/n_0) + \BO{n} & \textnormal{otherwise }.\\
\end{array}
\right.
$$
By solving the recurrence, we get 
$$T(n) = \BO{ \frac{n_0}{p_0}  \left( \frac{n}{m} \right)^{\omega_0}  \left( m \sqrt{n_0} + \ell \right)},$$ where
$\omega_0 = \log_{n_0}{p_0}$. Since $n_0$ and $p_0$ are independent of $n$, we get the claimed result.
\end{proof}

The standard recursive matrix multiplication algorithm gives  $\BO{n^{3/2}/m^{1/2} + (n/m)^{3/2}\ell)}$ time. 
With the Strassen algorithm, we get   $\BO{n^{1.4037}/m^{0.4037}+(n/m)^{1.4037}\ell}$ time 

We now show how to decrease the latency cost, i.e., $(n/m)^{3/2}\ell$, in the TCU algorithm based on the standard algorithm.
The idea is to keep as much as possible the right  matrix $B$ within the tensor unit by using a tall left matrix $A$.
We split the left matrix $A$ into $\sqrt{n/m}$ blocks $A_i$ of size $\sqrt{n}\times \sqrt{m}$ (i.e., vertical strips of width $\sqrt{m}$) and the right matrix $B$ into square blocks $B_{i,j}$ of size $\sqrt{m}\times \sqrt{m}$, with $0\leq i,j< \sqrt{n/m}$.
Then, we compute  $C_{i,j}= A_i \cdot B_{i,j}$ for each $0\leq i,j< \sqrt{n/m}$ using the tensor unit in time $\BO{\sqrt{nm}+\ell}$. 
The final matrix $C$ follows by computing the $\sqrt{n}\times \sqrt{m}$ matrices $C_{i} = \sum_{j=0}^{\sqrt{n/m}-1} C_{i,j}$.

\begin{theorem}\label{thm:mmult2}
There exists an algorithm that multiplies two $\sqrt{n}\times \sqrt{n}$ matrices in the $(m,\ell)$-TCU model in time
$$
T(n) = \BT{\frac{n^{3/2}}{m^{1/2}} + \frac{n}{m}\ell}.
$$
The algorithm is optimal when only semiring operations are allowed.
\end{theorem}
\begin{proof}
Each multiplication $C_{i,j}= A_i \cdot B_{i,j}$ requires $\BO{n\sqrt{m} + \ell}$ time. Since  there are ${n/m}$ such multiplications, the upper bound follows. The cost of the final summation is negligible.

When using only semiring operations, any algorithm must compute $n^{3/2}$ elementary products. Since each call to a tensor computes $m^{3/2}$ elementary products in $\BT{m}$ time using a systolic algorithm, we need $\BOM{n^{3/2}/m^{1/2}}$ time.
Furthermore, since all entries of $B$ must be loaded in the tensor unit at least once and we cannot load more than $m$ entries in $B$ per tensor operation, the algorithm has to load at least $n/m$ distinct right matrices in the tensor unit; then a $\BOM{\ell n/m}$ lower bound on the time also follows.
\end{proof}

From the previous Theorem~\ref{thm:mmult2}, we get the following corollary for rectangular matrices (a similar result holds also when using the algorithm for fast matrix multiplication in Theorem~\ref{thm:mmult}).
\begin{corollary}\label{cor:mmult2}
A $\sqrt{n} \times r$ matrix can be multiplied by an $r \times \sqrt{n}$ matrix in the $(m,\ell)$-TCU model in time
$$
T(n, r) = \BT{\frac{r n}{m^{1/2}} + \frac{r \sqrt{n}}{m}\ell},
$$
assuming $n, r^2 \geq m$.
\end{corollary}
\begin{proof}
It suffices to decompose the problem into products of size $s\times s$ where $s=\min\{\sqrt{n}, r\}$ and then apply Theorem \ref{thm:mmult2}.
\end{proof}

\subsubsection*{Sparse matrix multiplication}
A TCU algorithm to multiply two sparse matrices follows from the work~\cite{Riko15} that uses as a black box a fast matrix multiplication algorithm for multiplying two $\sqrt{n}\times \sqrt{n}$ matrices in $\BO{n^{\omega/2}}$ time.
Let $I$ be the number of non-zero entries in the input matrices $A$ and $B$, and let $Z$ be the number of non-zero entries in the output $C=A\cdot B$. 
We consider here the case  where the output is balanced, that is there are $\BT{Z/\sqrt{n}}$ non-zero entries per row or column in $C$; the more general case where non-zero entries are not balanced is also studied in~\cite{Riko15} and can be adapted to TCU with  a similar argument.
The algorithm in~\cite{Riko15} computes the output in time $\BTO{\sqrt{n}Z^{(\omega-1)/2}+I}$ with high probability.
The idea is to compress the rows of $A$ and the column of $B$ from $\sqrt{n}$ to $\sqrt{Z}$ using a hash function or another compression algorithm able to build a re-ordering of the matrix $A$.
Then the algorithm computes a dense matrix product between a $\sqrt{Z}\times \sqrt{n}$ matrix and a $\sqrt{n}\times \sqrt{Z}$ using the fast matrix multiplication algorithm.
By replacing the fast matrix multiplication with the result of Theorem~\ref{thm:mmult}, we get the following.

\begin{theorem}
Let $A$ and $B$ be two sparse input matrices of size $\sqrt{n}\times \sqrt{n}$ with at most $I$ non-zero entries, and assume that $C=A   \cdot B$ has at most $Z$ non-zero entries  evenly balanced among rows and columns.
Then there exists an algorithm for the $(m,\ell)$-TCU model requiring time:
$$
T(n,Z,I)=\BO{\sqrt{\frac{n}{Z}}\left(\frac{Z}{m}\right)^{\omega_0}(m+\ell)+I},
$$
when $Z\geq m$ and where $\omega_0 = \log_{n_0}p_0$ is the exponent given by a Strassen-like algorithm.
\end{theorem}
\begin{proof}
The cost is dominated by the  matrix product between a $\sqrt{Z}\times \sqrt{n}$ matrix and a $\sqrt{n}\times \sqrt{Z}$ using the fast matrix multiplication algorithm in Theorem~\ref{thm:mmult},
\end{proof}

 \subsection{Gaussian Elimination without Pivoting}
 \label{ssec:geplu}

 \begin{figure}[t]
 
 \begin{center}
 \scalebox{0.95}{
 \framebox{
 \begin{minipage}{2.8in}
 {\footnotesize
 \noindent
 \begin{itemize}
 
 \item[{\bf{1.}}] \xfor $k \leftarrow 1$ \xto $\sqrt{n} - 2$ \xdo
 
 \item[{\bf{2.}}] \T \xfor $i \leftarrow k + 1$ \xto $\sqrt{n} - 1$ \xdo
 
 \item[{\bf{3.}}] \T\T \xfor $j \leftarrow k + 1$ \xto $\sqrt{n}$ \xdo
 
 \item[{\bf{4.}}] \T\T\T $c[i, j] \leftarrow c[i, j] + \left( - {{c[i, k]} \over {c[k, k]}} \right) \times c[k, j]$
 
 \end{itemize}
 }
 \end{minipage}
 }
 }
 \end{center}
 
 \caption{Forward phase of Gaussian elimination.}
 \label{fig:gep}
 \vspace{-0.2cm}
 \end{figure}

\begin{figure}[t]
\centering
\includegraphics[scale=0.45]{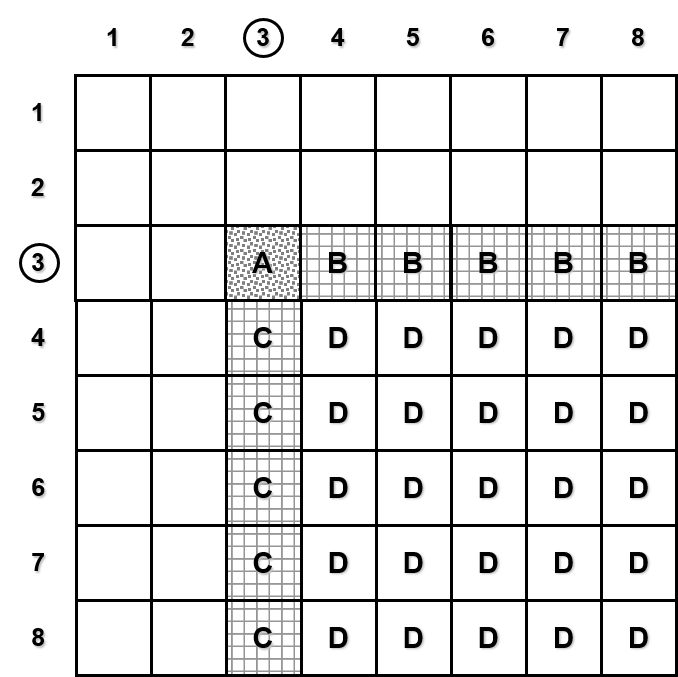}
    \caption{Blocked version of Gaussian elimination.
 It shows which functions update which blocks in iteration $k = 3$
of the outermost \xfor loop of \func{GE-forward}.} 
    \label{fig:gep-tpu}
\end{figure}

 \begin{figure*}[t!] 

 \begin{minipage}{6.6in}
 \begin{center}
 \begin{minipage}{3in}
% \vspace{-0.1cm}
 \framebox{
 \begin{minipage}{2.8in}
 {\scriptsize
 \vspace{-0.2cm}
 \medskip\noindent\func{GE-forward$(~X~)$}

 \vspace{0.1cm}
 \noindent
 ($X$ points to the $\sqrt{n} \times \sqrt{n}$ 
  input matrix $c$. We assume that $m$ divides $n$, where 
  $\sqrt{m} \times \sqrt{m}$ is the size of the matrix multiplication
  unit of the TCU.)

 \vspace{0.1cm}
 \noindent
 \begin{enumerate}

 \nvs
 \item Split $X$ into $\sqrt{n \over m} \times \sqrt{n \over m}$ square submatrices 
       of size $\sqrt{m} \times \sqrt{m}$ each. The submatrix of $X$
       at the $i$-th position from the top and the $j$-th 
       position from the left is denoted by $X_{ij}$.
       $X'$ is a $\sqrt{m} \times \sqrt{n}$ matrix
       split into $\sqrt{m} \times \sqrt{m}$ submatrices,
       where the submatrix at $j$-th position from the
       left is denoted by $X'_j$. 

 \nvs
 \item \xfor $k \leftarrow 1$ \xto $\sqrt{n/m}$ \xdo \label{ln:k}

 \nvs
 \item \T \func{A$(~X_{kk}~)$} \label{ln:A}

 \nvs
 \item \T \xfor $j \leftarrow k + 1$ \xto $\sqrt{n/m}$ \xdo \label{ln:B}

 \nvs
 \item \T\T   \func{B$(~X_{kj},~ X_{kk},~ X'_j~)$}

 \nvs
 \item \T \xfor $i \leftarrow k + 1$ \xto $\sqrt{n/m}$ \xdo \label{ln:C}

 \nvs
 \item \T\T   \func{C$(~X_{ik},~ X_{kk}~)$}

 \nvs
 \item \T \xfor $j \leftarrow k + 1$ \xto $\sqrt{n/m}$ \xdo \label{ln:D-j}

 \nvs
 \item \T\T \xfor $i \leftarrow k + 1$ \xto $\sqrt{n/m}$ \xdo \label{ln:D-i}

 \nvs
 \item \T\T\T \func{D$(~X_{ij},~ X_{ik},~ X'_j~)$} \label{ln:D-call}

 \end{enumerate}
 %\medskip\noindent{\func{A Ends}}
 }
 \end{minipage}
 }

 \framebox{
 \begin{minipage}{2.8in}
 {\scriptsize
 \vspace{-0.2cm}
 \medskip\noindent\func{D$(~X,~Y,~Z~)$}

 \vspace{0.1cm}
 \noindent
 ($X$, $Y$ and $Z$ point to disjoint $\sqrt{m} \times \sqrt{m}$ 
  matrices, where $\sqrt{m} \times \sqrt{m}$ is 
  the size of the matrix multiplication unit of the TCU.)

 \vspace{0.1cm}
 \noindent
 \begin{enumerate}

 \nvs
 \item \xfor $k \leftarrow 1$ \xto $\sqrt{m}$ \xdo
 
 \nvs
 \item \T \xfor $i \leftarrow 1$ \xto $\sqrt{m}$ \xdo
 
 \nvs
 \item \T\T \xfor $j \leftarrow 1$ \xto $\sqrt{m}$ \xdo
 
 \nvs
 \item \T\T\T $X[i, j] \leftarrow X[i, j] + Y[i, k] \times Z[k, j]$

 \end{enumerate}
 %\medskip\noindent{\func{A Ends}}
 }
 \end{minipage}
 }

 \end{minipage}
 \begin{minipage}{3.2in}
% \vspace{-0.1cm}
 \framebox{
 \begin{minipage}{3.1in}
 {\scriptsize
 \vspace{-0.1cm}
 \medskip\noindent\func{A$(~X~)$}

 \vspace{0.1cm}
 \noindent
 ($X$ points to a $\sqrt{m} \times \sqrt{m}$ 
  matrix, where $\sqrt{m} \times \sqrt{m}$ is 
  the size of the matrix multiplication unit 
  of the TCU.)

% \vspace{0.1cm}
 \noindent
 \begin{enumerate}

 \nvs
 \item \xfor $k \leftarrow 1$ \xto $\sqrt{m} - 1$ \xdo
 
 \nvs
 \item \T \xfor $i \leftarrow k + 1$ \xto $\sqrt{m}$ \xdo
 
 \nvs
 \item \T\T \xfor $j \leftarrow k + 1$ \xto $\sqrt{m}$ \xdo
 
 \nvs
 \item \T\T\T $X[i, j] \leftarrow X[i, j] - { \left( {X[i, k] \times X[k, j] } \right) / {X[k, k]} }$

 \end{enumerate}
 %\medskip\noindent{\func{A Ends}}
 }
 \end{minipage}
 }

 \framebox{
 \begin{minipage}{3.1in}
 {\scriptsize
 \vspace{-0.1cm}
 \medskip\noindent\func{B$(~X,~Y,~X'~)$}

 \vspace{0.1cm}
 \noindent
 ($X$, $Y$ and $X'$ point to disjoint $\sqrt{m} \times \sqrt{m}$ 
  matrices, where $\sqrt{m} \times \sqrt{m}$ is 
  the size of the matrix multiplication unit of the TCU.)

% \vspace{0.1cm}
 \noindent
 \begin{enumerate}

 \nvs
 \item \xfor $k \leftarrow 1$ \xto $\sqrt{m} - 1$ \xdo
 
 \nvs
 \item \T \xfor $i \leftarrow k + 1$ \xto $\sqrt{m}$ \xdo
 
 \nvs
 \item \T\T \xfor $j \leftarrow 1$ \xto $\sqrt{m}$ \xdo
 
 \nvs
 \item \T\T\T $X[i, j] \leftarrow X[i, j] - { \left( {Y[i, k] \times X[k, j] } \right) / {Y[k, k]} }$

 \nvs
 \item \xfor $i \leftarrow 1$ \xto $\sqrt{m}$ \xdo
 
 \nvs
 \item \T \xfor $j \leftarrow 1$ \xto $\sqrt{m}$ \xdo
 
 \nvs
 \item \T\T $X'[i, j] \leftarrow - X[i, j] / Y[i, i]$

 \end{enumerate}
 %\medskip\noindent{\func{A Ends}}
 }
 \end{minipage}
 }

 \framebox{
 \begin{minipage}{3.1in}
 {\scriptsize
 \vspace{-0.1cm}
 \medskip\noindent\func{C$(~X,~Y~)$}

 \vspace{0.1cm}
 \noindent
 ($X$ and $Y$ point to disjoint $\sqrt{m} \times \sqrt{m}$ 
  matrices, where $\sqrt{m} \times \sqrt{m}$ is 
  the size of the matrix multiplication unit of the TCU.)

% \vspace{0.1cm}
 \noindent
 \begin{enumerate}

 \nvs
 \item \xfor $k \leftarrow 1$ \xto $\sqrt{m}$ \xdo
 
 \nvs
 \item \T \xfor $i \leftarrow 1$ \xto $\sqrt{m}$ \xdo
 
 \nvs
 \item \T\T \xfor $j \leftarrow k + 1$ \xto $\sqrt{m}$ \xdo
 
 \nvs
 \item \T\T\T $X[i, j] \leftarrow X[i, j] - { \left( {X[i, k] \times Y[k, j] } \right) / {Y[k, k]} }$

 \end{enumerate}
 %\medskip\noindent{\func{A Ends}}
 }
 \end{minipage}
 }

 \end{minipage}

 \vspace{-0.2cm}
 \caption{TCU algorithm for Gaussian elimination without pivoting which is called as 
 \func{GE-forward}$(~c~)$, where $c$ is the $\sqrt{n} \times \sqrt{n}$ matrix 
 representing a system of $\sqrt{n} - 1$ equations with $\sqrt{n} - 1$ unknowns.}
\label{fi:gep-tcu}

 \end{center}
 \end{minipage}
 \end{figure*}

 Gaussian elimination
 without pivoting is used in the solution of systems of linear
 equations and LU decomposition of symmetric positive-definite or
 diagonally dominant real matrices \cite{CLRS01}. We represent a
 system of $r - 1$ equations in $r - 1$ unknowns ($x_{1}, x_{2},
 \ldots, x_{r - 1}$) using an $r \times r$ matrix $c$,
 where the $i$'th ($1 \leq i < r$) row
 represents the equation $a_{i, 1}x_{1} + a_{i, 2}x_{2} + \ldots
 + a_{i, r - 1}x_{r - 1} = b_{i}$:
\begin{displaymath}
 c = \left( \begin{array}{ccccc}
              a_{1, 1} & a_{1, 2} & \ldots & a_{1, \sqrt{n} - 1} & b_{1}\\
              a_{2, 1} & a_{2, 2} & \ldots & a_{2, \sqrt{n} - 1} & b_{2}\\
              \vdots & \vdots & \ddots & \vdots & \vdots\\
              a_{\sqrt{n} - 1, 1} & a_{\sqrt{n} - 1, 2} & \ldots & a_{\sqrt{n} - 1, \sqrt{n} - 1} & b_{\sqrt{n} - 1}\\
              0 & 0 & \ldots & 0 & 0
             \end{array} \right)
 \end{displaymath} 
 
  The method proceeds in two phases.
 In the first phase, an upper triangular matrix is constructed from $c$
 by successive elimination of variables from the 
 equations. This phase requires $\BT{r^3}$ time (see code in Figure~\ref{fig:gep}). 
 In the second phase, the values of
 the unknowns are determined from this %upper triangular
 matrix by back substitution.
 It is straightforward to implement this
 second phase in $\BT{r^2}$ time, so we will concentrate on the first phase.

 Our TCU algorithm for the forward phase of Gaussian elimination
 without pivoting is shown in Figure \ref{fi:gep-tcu}.
 The algorithm is invoked as \func{GE-forward}$(~c~)$, where $c$ is the 
 $\sqrt{n} \times \sqrt{n}$ matrix representing a system of 
 $\sqrt{n} - 1$ equations with $\sqrt{n} - 1$ unknowns.
 Figure \ref{fig:gep-tpu} shows which function $\in \left\{ \func{A}, \func{B}, \func{C}, \func{D} \right\}$ 
 update which block in iteration $k = 3$ of the outermost 
 \xfor loop of \func{GE-forward}.
 In this algorithm only the calls to function \func{D}
 (in line \ref{ln:D-call}) which multiplies $\sqrt{m} \times \sqrt{m}$ matrices
 are executed on the TCU. In each iteration of the loop
 in line \ref{ln:D-j}, $X'_j$ is loaded into the TCU
 as the weight matrix, and the $\left( \sqrt{n/m} - k \right) \sqrt{m}
 = \sqrt{n} - k\sqrt{m}$ rows of the $\sqrt{n/m} - k$ submatrices $X_{ik}$
 inside the loop in line \ref{ln:D-i} are streamed through the TCU.

\begin{theorem}\label{thm:geplu}
The forward phase of Gaussian elimination without pivoting applied
on a system of $\sqrt{n} - 1$ equations with $\sqrt{n} - 1$ unknowns
can be performed in the $(m,\ell)$-TCU model in time
$$
T(n) = \BT{\frac{n^{3/2}}{m^{1/2}} + \frac{n}{m}\ell + n\sqrt{m}}.
$$
This complexity reduces to the optimal cost of multiplying two dense
$\sqrt{n} \times \sqrt{n}$ matrices (see Theorem \ref{thm:mmult2}) 
when $\sqrt{n} \geq m$.
\end{theorem}
\begin{proof}
The outermost loop of \func{GE-forward} in line \ref{ln:k}
is executed $\sqrt{n/m}$ times. In each iteration the cost
of executing line \ref{ln:A} is $\BT{m^{3/2}}$, and
that of executing the loops in lines \ref{ln:B} and \ref{ln:C}
is $\BT{m^{3/2}\left( \sqrt{n/m} - k \right) }$. Then in
each of the $\sqrt{n/m} - k$ iterations of the
loop in line \ref{ln:D-j} we load $X'_j$ as the
weight matrix into the TCU and then lines \ref{ln:D-i}--\ref{ln:D-call} 
are executed by streaming the $\left( \sqrt{n/m} - k \right)\sqrt{m} = \sqrt{n} - k\sqrt{m}$ rows
of $X_{k+1, k}, X_{k+2, k} \ldots X_{\sqrt{n/m}, k}$
through the TCU. Total cost of lines \ref{ln:D-j}--\ref{ln:D-call}
is thus $\BT{\left( \sqrt{n/m} - k \right) \left( \left( \sqrt{n} - k\sqrt{m} \right)\sqrt{m} + l\right)}$.
The total cost over all iterations of the loop
in line \ref{ln:k} is then 
$$\BT{\sum_{k=1}^{\sqrt{n/m}}{\left( m^{3/2} + \left( \sqrt{n/m} - k \right) 
\left( m^{3/2} + \left( \sqrt{n} - k\sqrt{m} \right)\sqrt{m} + l\right) \right)} }$$
from which we get $\BT{\frac{n^{3/2}}{m^{1/2}} + \frac{n}{m}\ell + n\sqrt{m}}.$
When $\sqrt{n} \geq m \Rightarrow \frac{n^{3/2}}{m^{1/2}} \geq n\sqrt{m}$,
the cost reduces to $\BT{\frac{n^{3/2}}{m^{1/2}} + \frac{n}{m}\ell}$
which matches the optimal cost of multiplying two dense
$\sqrt{n} \times \sqrt{n}$ matrices (see Theorem \ref{thm:mmult2}).
\end{proof}

%\subsection{Reachability and Shortest Distances}

 \subsection{Graph Transitive Closure}
 \label{ssec:ap-conn}
  
 \begin{figure}[ht]
 \begin{minipage}{3in}
 \begin{center}
 \scalebox{1}{
 \framebox{
 \begin{minipage}{2.9in}
 {\footnotesize
 \noindent
 \begin{itemize}
 
 \item[{\bf{1.}}] \xfor $k \leftarrow 1$ \xto $n$ \xdo
 
 \item[{\bf{2.}}] \T \xfor $i \leftarrow 1$ \xto $n$ \xdo
 
 \item[{\bf{3.}}] \T\T \xfor $j \leftarrow 1$ \xto $n$ \xdo
 
 \item[{\bf{4.}}] \T\T\T $d[i, j] \leftarrow d[i, j] \vee \left( d[i, k] \wedge d[k, j] \right)$
 
 \end{itemize}
 }
 \end{minipage}
 }
 }
 \end{center}
 \end{minipage}
 \vspace{-0.2cm}
 \caption{Iterative algorithm for computing the transitive closure of an $n$ vertex graph 
 given its $n \times n$ djacency matrix $d$ with $d[i, j] = 1$ if there is an edge from vertex $i$ 
 to vertex $j$ and $d[i, j] = 0$ otherwise.}
 \label{fig:tc}
 \vspace{-0.2cm}
 \end{figure}

\begin{figure}[ht]
\centering
\includegraphics[scale=0.45]{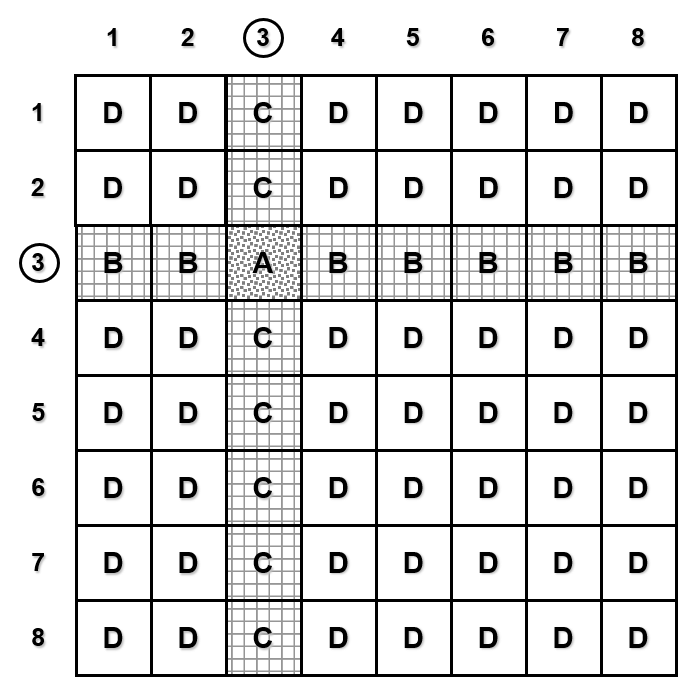}
    \caption{Blocked version of transitive closure. It shows
which functions update which blocks in iteration $k = 3$
of the outermost \xfor loop of \func{Transitive-Closure}.} 
    \label{fig:tc-tpu}
\end{figure}

 \begin{figure*}[ht!] 

 \begin{minipage}{6.6in}
 \begin{center}
 \begin{minipage}{3in}
% \vspace{-0.1cm}
 \framebox{
 \begin{minipage}{2.8in}
 {\scriptsize
 \vspace{-0.2cm}
 \medskip\noindent\func{Transitive-Closure$(~X~)$}

 \vspace{0.1cm}
 \noindent
 ($X$ points to the $n \times n$ 
  input $0/1$ matrix $d$. We assume that $m$ divides $n$, where 
  $\sqrt{m} \times \sqrt{m}$ is the size of the matrix multiplication
  unit of the TCU.)

 \vspace{0.1cm}
 \noindent
 \begin{enumerate}

 \nvs
 \item Split $X$ into ${n \over \sqrt{m}} \times {n \over \sqrt{m}}$ square submatrices 
       of size $\sqrt{m} \times \sqrt{m}$ each. The submatrix of $X$
       at the $i$-th position from the top and the $j$-th 
       position from the left is denoted by $X_{ij}$.

 \nvs
 \item \xfor $k \leftarrow 1$ \xto ${n \over \sqrt{m}}$ \xdo \label{ln:tc-k}

 \nvs
 \item \T \func{A$(~X_{kk}~)$} \label{ln:tc-A}

 \nvs
 \item \T \xfor $j \leftarrow 1$ \xto ${n \over \sqrt{m}}$ \xdo \label{ln:tc-B}

 \nvs
 \item \T\T \xif $j \neq k$ \xthen \func{B$(~X_{kj},~ X_{kk}~)$}

 \nvs
 \item \T \xfor $i \leftarrow 1$ \xto ${n \over \sqrt{m}}$ \xdo \label{ln:tc-C}

 \nvs
 \item \T\T \xif $i \neq k$ \xthen \func{C$(~X_{ik},~ X_{kk}~)$}

 \nvs
 \item \T \xfor $j \leftarrow 1$ \xto ${n \over \sqrt{m}}$ \xdo \label{ln:tc-D-j}

 \nvs
 \item \T\T \xfor $i \leftarrow 1$ \xto ${n \over \sqrt{m}}$ \xdo \label{ln:tc-D-i}

 \nvs
 \item \T\T\T \xif $i \neq k$ \xand $j \neq k$ \xthen  

 \nvs
 \item \T\T\T\T \func{D$(~X_{ij},~ X_{ik},~ X_{kj}~)$} \label{ln:tc-D-call}

 \end{enumerate}
 %\medskip\noindent{\func{A Ends}}
 }
 \end{minipage}
 }

 \framebox{
 \begin{minipage}{2.8in}
 {\scriptsize
 \vspace{-0.1cm}
 \medskip\noindent\func{A$(~X~)$}

 \vspace{0.1cm}
 \noindent
 ($X$ points to a $\sqrt{m} \times \sqrt{m}$ 
  $0/1$ matrix, where $\sqrt{m} \times \sqrt{m}$ is 
  the size of the TCU matrix multiplication unit.)

% \vspace{0.1cm}
 \noindent
 \begin{enumerate}

 \nvs
 \item \xfor $k \leftarrow 1$ \xto $\sqrt{m}$ \xdo
 
 \nvs
 \item \T \xfor $i \leftarrow 1$ \xto $\sqrt{m}$ \xdo
 
 \nvs
 \item \T\T \xfor $j \leftarrow 1$ \xto $\sqrt{m}$ \xdo
 
 \nvs
 \item \T\T\T $X[i, j] \leftarrow X[i, j] \vee { \left( {X[i, k] \wedge X[k, j] } \right) }$

 \end{enumerate}
 %\medskip\noindent{\func{A Ends}}
 }
 \end{minipage}
 }

 \end{minipage}
 \begin{minipage}{3.2in}
% \vspace{-0.1cm}
 \framebox{
 \begin{minipage}{3.1in}
 {\scriptsize
 \vspace{-0.1cm}
 \medskip\noindent\func{B$(~X,~Y~)$}

 \vspace{0.1cm}
 \noindent
 ($X$, $Y$ and $X'$ point to disjoint $\sqrt{m} \times \sqrt{m}$ 
  $0/1$ matrices, where $\sqrt{m} \times \sqrt{m}$ is 
  the size of the TCU matrix multiplication unit.)

% \vspace{0.2cm}
 \noindent
 \begin{enumerate}

 \nvs
 \item \xfor $k \leftarrow 1$ \xto $\sqrt{m}$ \xdo
 
 \nvs
 \item \T \xfor $i \leftarrow 1$ \xto $\sqrt{m}$ \xdo
 
 \nvs
 \item \T\T \xfor $j \leftarrow 1$ \xto $\sqrt{m}$ \xdo
 
 \nvs
 \item \T\T\T $X[i, j] \leftarrow X[i, j] \vee { \left( Y[i, k] \wedge X[k, j] \right) }$

 \end{enumerate}
 %\medskip\noindent{\func{A Ends}}
 }
 \end{minipage}
 }

 \framebox{
 \begin{minipage}{3.1in}
 {\scriptsize
 \vspace{-0.1cm}
 \medskip\noindent\func{C$(~X,~Y~)$}

 \vspace{0.1cm}
 \noindent
 ($X$ and $Y$ point to disjoint $\sqrt{m} \times \sqrt{m}$ 
  $0/1$ matrices, where $\sqrt{m} \times \sqrt{m}$ is 
  the size of the TCU matrix multiplication unit.)

% \vspace{0.2cm}
 \noindent
 \begin{enumerate}

 \nvs
 \item \xfor $k \leftarrow 1$ \xto $\sqrt{m}$ \xdo
 
 \nvs
 \item \T \xfor $i \leftarrow 1$ \xto $\sqrt{m}$ \xdo
 
 \nvs
 \item \T\T \xfor $j \leftarrow 1$ \xto $\sqrt{m}$ \xdo
 
 \nvs
 \item \T\T\T $X[i, j] \leftarrow X[i, j] \vee { \left( X[i, k] \wedge Y[k, j] \right) }$

 \end{enumerate}
 %\medskip\noindent{\func{A Ends}}
 }
 \end{minipage}
 }

 \framebox{
 \begin{minipage}{3.1in}
 {\scriptsize
 \vspace{-0.2cm}
 \medskip\noindent\func{D$(~X,~Y,~Z~)$}

 \vspace{0.1cm}
 \noindent
 ($X$, $Y$ and $Z$ point to disjoint $\sqrt{m} \times \sqrt{m}$ 
  $0/1$ matrices, where $\sqrt{m} \times \sqrt{m}$ is 
  the size of the TCU matrix multiplication unit.)

% \vspace{0.1cm}
 \noindent
 \begin{enumerate}

 \nvs
 \item \xfor $k \leftarrow 1$ \xto $\sqrt{m}$ \xdo \label{ln:tc-func-D-start-MM}
 
 \nvs
 \item \T \xfor $i \leftarrow 1$ \xto $\sqrt{m}$ \xdo
 
 \nvs
 \item \T\T \xfor $j \leftarrow 1$ \xto $\sqrt{m}$ \xdo
 
 \nvs
 \item \T\T\T $X[i, j] \leftarrow X[i, j] + \left( Y[i, k] \times Z[k, j] \right)$ \label{ln:tc-func-D-end-MM}

 \nvs
 \item \xfor $i \leftarrow 1$ \xto $\sqrt{m}$ \xdo
 
 \nvs
 \item \T \xfor $j \leftarrow 1$ \xto $\sqrt{m}$ \xdo
 
 \nvs
 \item \T\T \xif $X[i, j] > 1$ \xthen $X[i, j] \leftarrow 1$

 \end{enumerate}
 %\medskip\noindent{\func{A Ends}}
 }
 \end{minipage}
 }

 \end{minipage}

 \vspace{-0.2cm}
 \caption{TCU algorithm for computing transitive closure of an $n$-vertex graph which is called as 
 \func{Transitive-Closure}$(~d~)$, where $d$ is the $n \times n$ 
 adjacency matrix of the graph with $d[i, j] = 1$ if vertices $i$ and $j$ are
 adjacent and $d[i, j] = 0$ otherwise.}
\label{fi:tc-tcu}

 \end{center}
 \end{minipage}

 \end{figure*}

 For an $n$-vertex directed graph $G$, its {\em transitive closure}
 is given by an $n \times n$ matrix $c[1..n, 1..n]$, where for
 all $i, j \in [1, n]$, $c[i, j] = 1$ provided vertex $j$ is reachable
 from vertex $i$ and $c[i, j] = 0$ otherwise. Figure \ref{fig:tc}
 shows how to compute $c$ in $\BT{n^3}$ time by updating the $0/1$ adjacency
 matrix $d[1..n, 1..n]$ of $G$ in place. The algorithm is similar
 to the standard iterative matrix multiplication algorithm except
 that bitwise-AND ($\wedge$) and bitwise-OR ($\vee$)
 replace multiplication ($\times$) and addition ($+$),
 respectively.

 Figure \ref{fi:tc-tcu} shows the blocked version of
 the algorithm given in Figure \ref{fig:tc}. 
 Figure \ref{fig:tc-tpu} shows which function $\in \left\{ \func{A}, \func{B}, \func{C}, \func{D} \right\}$ 
 updates which block in iteration $k = 3$ of the outermost 
 \xfor loop of \func{Transitive-Closure}.
 However, we observe that function $\func{D}$ which updates
 block $X$ using data from blocks $Y$ and
 $Z$ that are disjoint from $X$ can be
 implemented to use ``$\times$'' and ``$+$''
 instead of ``$\wedge$'' and ``$\vee$'',
 respectively, provided we set $X[i, j] \leftarrow \min{\left( X[i, j], 1 \right)}$
 for all $i, j$ after it completes updating $X$.
 Function $\func{D}$ is invoked
 in line \ref{ln:tc-D-call} of \func{Transitive-Closure}
 almost $\frac{n^2}{m}$ times. We execute lines
 \ref{ln:tc-func-D-start-MM}-- \ref{ln:tc-func-D-end-MM}
 of function $\func{D}$ (which represent standard
 multiplication of two $\sqrt{m} \times \sqrt{m}$ matrices) 
 on a TCU. In each iteration of the loop
 in line \ref{ln:tc-D-j}, $X_{kj}$ is loaded into the TCU
 as the weight matrix, and the $\left( n/\sqrt{m} - 1 \right) \sqrt{m}
 = \sqrt{n} - \sqrt{m}$ rows of the $n/\sqrt{m} - 1$ submatrices $X_{ik}$
 inside the loop in line \ref{ln:tc-D-i} are streamed through the TCU.

\begin{theorem}\label{thm:tc}
The transitive closure of an $n$-vertex directed graph
can be computed in the $(m,\ell)$-TCU model in time
$$
T(n) = \BT{\frac{n^{3}}{\sqrt{m}} + \frac{n^2}{m}\ell + n^2\sqrt{m}}.
$$
This complexity reduces to the optimal cost of multiplying two dense
$n \times n$ matrices (see Theorem \ref{thm:mmult2}) 
when $n \geq m$.
\end{theorem}
\begin{proof}
The proof is similar to the proof of Theorem \ref{thm:geplu}
and is obtained by analyzing the steps of \func{Transitive-Closure}
shown in Figure \ref{fi:tc-tcu}.
\end{proof}

 \subsection{All Pairs Shortest Distances (APSD)}
 \label{ssec:ap-sp}

We discuss TCU implementation of Seidel's algorithm \cite{Seidel1995all} for computing APSD in an unweighted undirected graph $G = (V, E)$, where $n = |V|$ and vertices are numbered by unique integers from 1 to $n$. 

Let $A$ be the adjacency matrix of $G$. The adjacency matrix $A^{(2)}$ of the squared graph $G^{(2)} = (V, E^{(2)})$ is obtained by squaring $A$ and replacing all nonzero entries in the square matrix by 1. Indeed, for any given pair of vertices $u, v \in V$, $(u, v) \in E^{(2)}$ (i.e., $A^{(2)}[u, v] = 1$) provided there exists a vertex $w \in V$ such that $(u, w), (w, v) \in E$ (i.e., $A[u, w] = A[w, v] = 1$). Let $\delta(u, v)$ and $\delta^{(2)}(u, v)$ represent the shortest distance from $u$ to $v$ in $G$ and $G^{(2)}$, respectively. Seidel shows that if all $\delta^{(2)}$ values are known one can correctly compute all $\delta(u, v)$ values from them. Let $D^{(2)}$ be the distance matrix of $G^{(2)}$ and let $C = D^{(2)} A$. Then Seidel shows that for any pair $u, v \in V$, $\delta(u, v) = 2 \delta^{(2)}(u, v)$ provided $\sum_{(w, v) \in E}{D^{(2)}[u, w]} = C[ u, v ] \geq deg(v) \times D^{(2)}[u, v]$, and $\delta(u, v) = 2 \delta^{(2)}(u, v) - 1$ otherwise, where $deg(v)$ is the number of neighbors of $v$ in $G$. Thus the distance matrix $D$ of $G$ can be computed from $D^{(2)}$ by computing $C = D^{(2)} A$. The $D^{(2)}$ matrix is computed recursively. The base case is reached when we encounter $G^{(h)}$ where $h = \lceil{ \log_{2}(n) } \rceil$. It's adjacency matrix $A^{(h)}$ has all $1$'s, and it's distance matrix is simply $D^{(h)} = A^{(h)} - I_{n}$. Clearly, there are $h$ levels of recursion and in each level we compute two products of two $n \times n$ matrices. Hence, using Theorem \ref{thm:mmult} we obtain the following.

\begin{theorem}\label{thm:exact-apsd-uu}
All pairs shortest distances of an $n$-vertex
unweighted undirected graph can be computed
in the $(m,\ell)$-TCU model in time
$$
T(n) = \BO{\left({n^2}/{m}\right)^{\omega_0}(m+\ell) \log{n}}.
$$
\end{theorem}
\begin{proof}
The claim follows since we have $\BO{\log n}$ matrix multiplications of size $n\times n$, each one requiring $\BO{\left({n^2}/{m}\right)^{\omega_0}(m+\ell)}$ for Theorem \ref{thm:mmult}.
\end{proof}

\subsection{Discrete Fourier Transform}\label{sec:FFT}
The Discrete Fourier Transform $y$ of an $n$-dimensional (column) vector $x$ can be defined as the matrix-vector product $y = x^T\cdot W$, where $W$ is the Fourier matrix (or DFT matrix) and $T$ denotes the transpose of a matrix/vector. 
The Fourier matrix $W$ is a symmetric $n\times n$ matrix where the entry at row $r$ and column $c$ is defined as: $W_{r,c}=e^{-(2\pi i/n)rc}$.
%Therefore the DFT of $n$ vectors $x_i,\ldots x_n$ can be computed with the matrix multiplication $X^T \cdot W$, where the $i$-th column of $X$ denotes the $i$-th vector.

The Cooley-Tukey algorithm is an efficient and recursive algorithm for computing the DFT of a vector. 
The algorithm arranges $x$ as an $n_1\times n_2$ matrix $X$ (in row-major order) where $n=n_1\cdot n_2$;  each column $X_{*,c}$ is replaced with its DFT and then each entry $X_{r,c}$ is multiplied by the twiddle factor $w_n^{rc}$; finally, each row $X_{r,*}$ is replaced by its DFT and the DFT of $x$ is given by reading the final matrix $X$ in column-major order.

For simplicity, we assume that the TCU model can perform operations (e.g., addition, products) on complex numbers; this assumption can be easily removed with a constant slow down in the running time: for instance, the multiplication between $\sqrt{m}\times \sqrt{m}$ complex matrices can be computed with four matrix multiplications and two sums  of real values. 

To compute the DFT of $x$ using a $(m,\ell)$-TCU, we use the  Cooley-Tukey algorithm where we set $n_1=\sqrt{m}$ and $n_2 =  n/\sqrt{m}$ (we assume all values to be integers). Then, we use the tensor unit for computing the $n_2$ DFTs of size $n_1=\sqrt{m}$ by computing  $X^T\cdot W_{\sqrt{m}}$.
Then, we multiply each element in $X$ by its twiddle factor and transpose $X$. 
Finally, we compute the $n_1$ DFTs of size $n_2$: 
if $n_2 > \sqrt{m}$, the DFTs are recursively computed; otherwise, if $n_2 \leq \sqrt{m}$, the $n_1$ DFTs are computed with the multiplication $X^T \cdot W_{{n_2}}$ by using the tensor unit.  

\begin{theorem}\label{thm:dft}
The DFT of a vector with $n$ entries can be computed in a $(m,\ell)$-TCU in time
$
T(n) = \BO{(n + \ell) \log_m n}.
$
\end{theorem}
\begin{proof}
In each recursive level, we load matrix $B=W_{\sqrt{m}}$ at the beginning and then compute the $n_2$ DFTs of $n_1$ entries by multiplying an $n/\sqrt{m}\times \sqrt{m}$ matrix with $W_{\sqrt{m}}$ in time $\BO{n+\ell}$.
The remaining operations (i.e., matrix transposition and multiplication by twiddle factors) cost $\BO{n}$ time.
The running time is thus given by the next recurrence, from which follows the statement.
$$
T(n) = \left\{ 
\begin{array}{ll}
\sqrt{m} T(n/\sqrt{m}) + \BO{n + \ell} & \textnormal{if } n>m\\
\BO{m+\ell} &  \textnormal{if } n\leq m
\end{array}
\right.
$$
We observe that we use as base case $n\leq m$ and not $n\leq \sqrt{m}$: indeed, when $n\leq m$ we use the tensor core for computing the $n_2$ DFTs of size $n_1=\sqrt{m}$, but also the $n_1$ DFTs of size $n_2\leq \sqrt{m}$. 
This analysis gives a tighter upper bound than using $n\leq \sqrt{m}$ as base case.

%Finally, we note that that the right matrix is always $M$ is the same in each level, however we are not able to generate a tall matrix including several recursive levels since the input of the $i$-th level follows by a transposition of the matrix at level $i-1$.
\end{proof}

We observe that the above algorithm generalizes the approach used in~\cite{Sorna18} on an NVIDIA Verdi architecture. 
The paper decomposes the vector using $n_1=4$ and $n_2=n/4$ and then solves subproblems of size 4 using a tensor core, since a TC can multiply  $4\times 4$ matrices. 

\subsection{Stencil computations}
Stencil computations are  iterative kernels over a $d$-dimensional array, widely used in scientific computing. 
Given a $d$-dimensional matrix $A$, a stencil computation performs a sequence of sweeps over the input: in a sweep, each cell is updated with a function $f(\cdot)$ of the values of its neighboring cells at previous sweeps. 
An example of  stencil computation is the discretization of the 2D heat equation, where each entry at time $t$ is updated as follows:
\begin{align*}
&A_t[x,y] =  A_{t-1}[x,y] + \\ 
&+\frac{\alpha \Delta t}{\Delta x^2}(A_{t-1}[x-1,y]+A_{t-1}[x+1,y]-2A_{t-1}[x,y])  \\
&+\frac{\alpha \Delta t}{\Delta y^2}(A_{t-1}[x,y-1]+A_{t-1}[x,y+1]-2A_{t-1}[x,y])
\end{align*}
where $\alpha, \Delta t, \Delta x^2, \Delta y^2$ are suitable constant values given by the heat diffusion equations and by the discretization step.

For the sake of simplicity, we assume $d=2$ and that each update depends only on the values of the cell and of its eight (vertical/horizontal/diagonal) neighbors at previous sweep.
However, all techniques presented in this section extend to any $d=\BO{1}$ and to any update function that depends on a constant number of neighbors.

Given  $n,k\geq 1$, an \emph{$(n,k)$-stencil computation}, over an input $\sqrt{n}\times\sqrt{n}$ matrix $A$ is the matrix $A_k$ obtained by the following iterative process: let $A_0=A$ and $1\leq t \leq k$; matrix $A_t$ is defined by computing, for each $0\leq i,j<\sqrt{n}$, 
$A_t[i,j]=f(i,j,A_{t-1})$
where $f$ is a suitable function of cells $A_{t-1}[i+\alpha,j+\beta]$ with  $\alpha,\beta \in \{-1,0,1\}$.
We say that a stencil computation is \emph{linear} if $f$ is a linear, that is
$$A_t[i,j] = \sum_{\alpha,\beta \in \{-1,0,1\}} w_{\alpha,\beta} A_{t-1}[i+\alpha,j+\beta].$$
where $w_{\alpha,\beta}$ are suitable real values. 
The above stencil computation for approximating heat equations is linear. 
We assume $k$ to be even and that all values are integers.

By unrolling the update function of a linear $(n,k)$-stencil computation, each entry $A_k[i,j]$ can be represented as a linear combination of $\BO{k^2}$ entries of $A$, specifically all entries $(i',j')$ in $A$ where $|i-i'|\leq k$ and $|j-j'|\leq k$. 
That is, there exists a $(2k+1) \times (2k+1)$  matrix $W$ such that
$$A_t[i,j]  = \sum_{-k\leq \alpha,\beta\leq k} W[k+\alpha, k+\beta] A[i+\alpha, j+\beta].$$
%We observe that, with a slight abuse of notation, matrix $A_t$ is given by the circular discrete convolution of $A_0$ and $W$.  

We now show that a linear $(n,k)$-stencil on a matrix $A$ reduces to $\BT{n/k^2}$ convolutions of size $\BO{k^2}$, which are then computed with the TCU algorithm for DFT in Theorem~\ref{thm:dft}.
Let us assume that matrix $A$ is split into submatrices $A_{r,c}$ of size $k\times k$, with $0\leq r,c < \sqrt{n}/k$; similarly, let $A_{k,r,c}$ denote the $k\times k$ submatrices of $A_{k}$.
For each $A_{r,c}$, we define the following matrix $A'_{r,c}$ of size  $3k\times 3k$:
$$
A'_{r,c} = \left[
\begin{array}{lll}
A_{r-1, c-1} & A_{r-1, c} & A_{r-1, c+1}\\
A_{r, c-1} & A_{r, c} & A_{r, c+1}\\
A_{r+1, c-1} & A_{r+1, c} & A_{r+1, c+1}
\end{array}
\right].
$$ 
where we assume that a matrix $A_{i,j}$ is a zero matrix when $i$ and $j$ are not in the range $[0, \sqrt{n}/k)$.
We then compute the circular discrete convolution $A^*_{r,c} = A'_{r,c} \circledast W'$, where $W'$ is a $3k\times 3k$ matrix obtained by
 flipping $W$ and by adding $k/2$ (resp., $k/2-1$) rows and columns of zeros on the left and top (resp., right and bottom) sides of $W$.\footnote{With a slight abuse of notation, given two $n\times n$ matrices $A$ and $B$ with $n$ even, we define $(A \circledast B)[i,j] = \sum_{\alpha,\beta \in [-n/2,n/2)} A[(i+\alpha)\mod n,(j+\beta)\mod n]W[n/2-\alpha, n/2-\beta]$. In the paper, we omit the mod operation from the notation.}
Finally, we set  $A_{k,r,c}$ to be the $k\times k$ matrix obtained from $A^*_{r,c}$  by selecting the $i$-row and $j$-th column for all  $k\leq i,j < 2k$. 
By repeating the following procedure for each submatrix $A_{r,c}$, we get the output matrix $A_k$.

Each convolution can be efficiently computed by exploiting the convolution theorem and the DFT algorithm of Theorem~\ref{thm:dft}. 
We indeed recall that a 2-dimensional DFT is given by computing a 1-dimensional DFT for each row and for each column.
If $W$ is given, we have the following:

\begin{lemma}\label{lem:stencil}
Given a linear $(n,k)$-stencil computation and its weight matrix $W$, then the stencil can be computed in a $(m,\ell)$-TCU in time
$
T(n,t) = \BO{(n + \ell) \log_m k}. 
$
\end{lemma}
\begin{proof}
The correctness follows by observing that, by definition of convolution, we have for each  $k\leq i,j < 2k$: 
\begin{align*}A^*_{r,c}[i,j] &= \hspace{-1em}\sum_{\substack{\alpha,\beta\in \\ [-\frac{3k}{2},\frac{3k}{2})}}  A'_{r,c}[i+\alpha,j+\beta] W'\left[\frac{3k}{2}-\alpha,\frac{3k}{2}-\beta\right]\\ 
 &=\sum_{\substack{\alpha,\beta\in \\ [-k,k]}} A'_{r,c}[i+\alpha,j+\beta] W\left[k+\alpha,k+\beta\right] 
\\&= A_k[i,j].
\end{align*}

Each convolution has size $\BT{k^2}$ and is computed with $\BT{k}$  DFT of $\BT{k}$ elements and $\BT{k^2}$ element-wise products. 
The time per convolution is then $\BO{(k^2+\ell k)\log_m k}$ time, while the  time of the entire algorithm is $\BO{(n+\ell n/k)\log_m k}$. 
However, the latency cost can be reduced by observing that all the $\BT{n/k}$ DFTs of $\BT{k}$ elements can be computed concurrently; 
then, for each of the $\BO{\log_m k}$ recursive levels of the DFT algorithm, we use a 2 tall left matrices of size $\BT{n/k}\times \BT{k}$ for computing the $\BO{n/k}$ DFTs (we need two matrices because we first compute the DFTs of the rows, and then on the columns).
The claimed result follows.  
\end{proof}

The weight matrix can be trivially computed in $\BO{k^3}$ time by recursively unrolling function $f$. 
However, as soon as $k\geq (n\log_m k)^{1/3}$, the cost for computing $W$ dominates the cost of the stencil algorithm.
A more efficient solution follows by representing $W$ as the powering of a  bivariate polynomial and then using the DFT to compute it, as shown in the following lemma.

\begin{lemma}\label{lem:weights}
The weight matrix of a linear $(n,k)$-stencil computation can be computed in a $(m,\ell)$-TCU in time $\BO{k^2 \log_m k + \ell \log k}$.
\end{lemma}
\begin{proof}
Let $f_{u}$ be function $f$ after unrolling it $u$-th times,
that is $A_k[i,j] = f_u(i,j,A_{k-u})$.
We assign to each $f_u$ a bivariate polynomial $F_u(x,y) = \sum_{-k\leq \alpha,\beta\leq k} w_u[\alpha,\beta] x^\alpha y^\beta$, where $w_u[\alpha,\beta]$ are the coefficients of $A_{k-u}[i+\alpha,j+\beta]$ in $f_u$.
By setting $F_0(x,y)= 1$, it follows by induction that $F_u(x,y) = F_{u-1}(x,y) P(x,y)$
where $P(x,y)=\sum_{\alpha,\beta\in\{-1,0,1\}} w_{\alpha,\beta} x^{\alpha}y^\beta$.
Entry $W[i,j]$ is then given by the coefficient of the $x^{i-k}y^{j-k}$ in $F_{k}=F_0(x,y)*P(x,y)^k$.
Since $P(x,y)^k = P(x,y)^{\lceil k/2 \rceil})^2 P(x,y)^{k \mod 2} $, we compute recursively the coefficient in $\BT{\log k}$ recursive calls, each one performing a convolution using the TCU algorithm for DFT of geometrically decreasing size.
We then get the main statement.
\end{proof}

By the previous two results, we then have the main result:
\begin{theorem}
Given a linear $(n,k)$-stencil computation with $k\leq n$, then the stencil can be computed in a $(m,\ell)$-TCU in time
$
T(n,t) = \BO{n \log_m k + \ell \log k}. 
$ 
\end{theorem}
\begin{proof}
The claim follows by summing the time complexities of Lemmas~\ref{lem:weights} and~\ref{lem:stencil}.
\end{proof}

\subsection{Integer multiplication}
We now study how to multiply two long integers by exploiting a $(m,\ell)$-TCU. 
The input is given by two integers $a$ and $b$ of $n$ bits each (without loss of generality, we assume both integers to be positive and $n>m$), and the output is the binary representation of $c=a*b$, of size  $2n-1$. 
For this problem, we introduce in the design a third parameter $\kappa$, which is the bit length of a memory word in the  TCU model.
We assume that $\kappa = \BOM{\log n}$, that is there are enough bits in a word to store the input/output size.
It is easy to see that the tensor unit can multiply  matrices of (positive) integers of  $\kappa'=\kappa/4$ bits without overflow: the largest integer in the output matrix using $\kappa'$ bits is $2^{2\kappa'}\sqrt{m}$ which requires $2\kappa'+\log \sqrt{m}<\kappa$ (if $n>>m$, then $\kappa'=\kappa/2  -1$ suffices).

We initially show how to speed up the long integer multiplication algorithm~\cite{KleinbergTardos16}, also known as the schoolbook algorithm, by exploiting the tensor unit.
Then, we will use this algorithm to improve the recursive Karatsuba algorithm~\cite{Karatsuba63}.

Let  $A(x) = \sum_{i=0}^{n'-1} A_i x^i$ be a polynomial where $n'=n/\kappa'$ and $A_i = (a_{(i+1)\kappa'-1}\ldots a_{i\kappa'})_2$ is the integer given by the $i$th segment of $\kappa'$ bits of $a$. Let $B(x)$  be defined similarly for $b$. 
We have that $a=A(2^{\kappa'})$ and $b=B(2^{\kappa'})$.
We define $C(x) = A(x)\cdot B(x)$ and we observe that $c$ is given by evaluating $C(2^{\kappa'})$. 
Note that $A(X)$ and $B(X)$ have degree $n'-1$, while $c$ has degree at most $(2n-1)/\kappa'\leq 2n'-1$.
The coefficients of $C(x)$ can be computed with the matrix multiplication $C=A\cdot B$ where: 1) $B$ is the column vector with the $n'$ coefficients of $B(X)$; 2) $A$ is a $(2n'-1)\times n'$ matrix where 
$A_{i,j}=A_{n'-i+j-1}$ and we assume that $A_h=0$ if $h<0$ or $h\geq n'$.

The product $C=A\cdot B$ cannot exploit TCU since $B$ is a vector. 
To fully exploit an $(m,\ell)$-TCU, we show how to calculate $C$ coefficients via the multiplication $C' = A'\cdot B'$ where $A$ is a $(n'+\sqrt{m}-1)\times \sqrt{m}$ matrix and $B$ is a $ \sqrt{m} \times n'/\sqrt{m}$ matrix.
\begin{itemize}
    \item Matrix $B'$ follows by considering vector $B$ as the column major representation of a $ \sqrt{m} \times n'/\sqrt{m}$ matrix, that is $B'_{i,j} = B_{n'-i-j\sqrt{m}-1}$.
    \item Matrix $A'$ is given by considering all segments  of length $\sqrt{m}$ in the sequence $0_{\sqrt{m}-1}, A_0, A_1, \ldots$ $ A_{n'-1},0_{\sqrt{m}-1}$, where $0_{\sqrt{m}-1}$ denotes a sequence of $\sqrt{m}-1$ zeros. 
More formally,  the $i$th row $A'_{i,*}$ is $[A_{n'-i-1},A_{n'-i-2},\ldots A_{n'-i-\sqrt{m}}]$,  where we assume again that $A_h=0$ if $h<0$ or $h\geq n'$.
\end{itemize}

Then, we compute $C' = A'\cdot B'$ with the algorithm for dense matrix multiplication of Theorem~\ref{thm:mmult2} (or equivalently the algorithms of Theorem~\ref{thm:mmult}): We decompose $B'$ into  into $n'/m$ submatrices of size $\sqrt{m}\times \sqrt{m}$ and then compute $n'/m$ products of a $(n'+\sqrt{m}-1)\times \sqrt{m}$ matrix with a $\sqrt{m}\times \sqrt{m}$ matrix. 
Finally, the coefficient of the $x^h$ indeterminate in $C(x)$, for each $0\leq h < 2n'-1$,  follows by summing all entries in $C'_{i,j}$ such that $h = 2(n'-1)-i-j\sqrt{m}$.
Finally we compute $c=C(2^{\kappa'})$.

\begin{theorem}\label{thm:intmul}
Two integers of $n$ bits can be multiplied in a $(m,\ell)$-TCU with $\kappa$ bits operations in time
$$
T(n)=\BO{\frac{n^2}{\kappa^2\sqrt{m}} + \frac{n}{\kappa m}\ell}.
$$
\end{theorem}
\begin{proof}
Let $C_h$ be the coefficient of $C(x)$ of  $x^h$  indeterminate. 
We have:
$$
C_h = \sum_{i,j| i+j=h} a_i b_j = 
\sum_{p=0}^{\lceil h/\sqrt{m}\rceil} \sum_{i=h-(p+1)\sqrt{m}+1}^{h-p\sqrt{m}} a_i b_{h-i}
$$
where as usual $a_i=0$ and $b_i=0$ if $i<0$ or $i\geq n'$.
By definition of matrices $A'$ and $B'$, we can rewrite the previous equation as:
\begin{align*}
C_h &= 
\sum_{p=0}^{\lceil h/\sqrt{m}\rceil}
A'_{n'-h+p\sqrt{m}-1,*} B'_{*,(n'-1)/\sqrt{m}-p} \\
&= \sum_{p=0}^{\lceil h/\sqrt{m}\rceil}
C_{n'-h+p\sqrt{m}-1,(n'-1)/\sqrt{m}-p}
\end{align*}
We observe that the last sum is including all entries in $C_{i,j}$ where $h = 2(n'-1)-i-j\sqrt{m}$, as required in the algorithm.
The algorithm then correctly computes all $C_h$ coefficients and hence $c=a\cdot b$.

We now consider the running time. 
The cost of the matrix multiplication $C'=A'\cdot B'$ is $\BO{\frac{n'}{m}\left(n'\sqrt{m}+\ell\right)}$. 
The cost of computing the $C_h$ coefficients and the final evaluation is upper bounded by $\BO{n'/\sqrt{m}}$. 
The claim follows.
\end{proof}

The Karatsuba algorithm is a well-known algorithm that computes $c=a\cdot b$ by recursively computing three integer multiplications  of size $n/2$ and then combining the solution in $\BO{n/\kappa}$ time. 
If we stop the recursion as soon as the input size is $n\leq k\sqrt{m}$ and solve the subproblem with the algorithm of Theorem~\ref{thm:intmul}, we get the following result.
\begin{theorem}
Two integers of $n$ bits can be multiplied in a $(m,\ell)$-TCU with $\kappa$ bits operations in time
$$
T(n)=\BO{\left(\frac{n}{\kappa \sqrt{m}}^{\log 3}\right) \left(\sqrt{m}+\frac{\ell}{\sqrt{m}}\right)}.
$$
\end{theorem}
\begin{proof}
The running time follows by the  recurrence:
$$
T(n) = \left\{ 
\begin{array}{ll}
3T(n/2) + \BO{n/\kappa}     &  \textnormal{if } n>k\sqrt{m}\\
\BO{\sqrt{m}+\ell/\sqrt{m}}     & \textnormal{if } n\leq k\sqrt{m}
\end{array}
\right.
$$
\end{proof}

\subsection{Batch polynomial evaluation}
We now show how to exploit the TCU for evaluating a given polynomial of $A(x) = \sum_{i=0}^{n-1} a_i x^i$ of degree $n-1$ on $p$ points $p_i$, with $0\leq i < p$. 
For simplicity we assume $n$ to be a multiple of $\sqrt{m}$, $p\geq \sqrt{m}$, and that the polynomial can be evaluated without overflow on the memory word available in the TCU model.

We initially compute for each $p_i$ the powers $p_i^0, p_i^1,\ldots p_i^{\sqrt{m}-1}$ and $p_i^{\sqrt{m}}, p_i^{2\sqrt{m}},\ldots p_i^{n-\sqrt{m}}$, that is  $p_i^j$ for each $j\in \{0,1,\ldots,\sqrt{m}-1\} \cup \{k \sqrt{m}, \forall k\in \{1,\ldots, n/\sqrt{m}-1\}\}$.
We define the following matrices:
\begin{itemize}
    \item A matrix $X$ of size $p\times \sqrt{m}$ where the $i$th row is $X_{i,*} = [p_i^0, p_i^1,\ldots, p_i^{\sqrt{m}-1}]$ for each $0\leq i <p$.
    \item A matrix $A$ of size $\sqrt{m}\times n/\sqrt{m}$ where $A_{i,j} = a_{i+j\sqrt{m}}$ for each $0\leq i <\sqrt{m}$ and $0\leq j<n/\sqrt{m}$. Stated differently, we consider the sequence $a_0,\ldots,a_{n-1}$ as the column major representation of $A$.
\end{itemize}
We then compute the product $C= X\cdot A$ by exploiting the tensor unit. As in the previous section, we decompose $A$ into $\sqrt{m}\times \sqrt{m}$ submatrices and then solve $n/m$ multiplications.
Then, for each $p_i$, the values $A(p_i)$ follows by the sum $\sum_{j=0}^{n/\sqrt{m}-1} C_{i,j} p^{j\sqrt{m}}_i$.

\begin{theorem}\label{thm:batch}
A polynomial of degree $n-1$ can be evaluated on $p$ points on a $(m,\ell)$-TCU with $\kappa$ bits operations in time
$$
T(n,p)=\BO{\frac{p n }{\sqrt{m}}+p\sqrt{m}+\frac{n}{m}\ell}.
$$
\end{theorem}
\begin{proof}
The correctness follows by observing that:
\begin{align*}
\sum_{j=0}^{n/\sqrt{m}-1} C_{i,j} x^{j\sqrt{m}} &= 
\sum_{j=0}^{n/\sqrt{m}-1}\left( \sum_{k=0}^{\sqrt{m}-1} p_i^k a_{k+j\sqrt{m}} \right) p_i^{j\sqrt{m}}\\
&=\sum_{j=0}^{n/\sqrt{m}-1}\sum_{k=0}^{\sqrt{m}-1} p_i^{k+j\sqrt{m}} a_{k+j\sqrt{m}} \\
&=\sum_{h=0}^{n-1} a_h p_i^h 
= A(p_i).
\end{align*}
The cost of the initial exponentiation and of the final evaluation is $\BO{p (\sqrt{m}+ n/\sqrt{m})}$.
The cost of computing $C$ is given by invoking the tensor unit $n/m$ times on two matrices of size $p\times \sqrt{m}$ and $\sqrt{m}\times \sqrt{m}$, that is $\BO{\frac{p n }{\sqrt{m}}+\frac{n}{m}\ell}$. 
\end{proof}

\section{Relation with the external memory model}\label{sec:emem}
In this section we highlight a relation between the external memory model and the TCU model.
We recall that the external memory model (also named I/O model and almost equivalent to the ideal cache model) is a model capturing the memory hierarchy and it consists of an external memory of potential unbounded size, of an internal memory of $M\geq 1$ words, and a processor. 
The processor can only perform operations with data in the internal memory, and moves (input/output) blocks of $B\geq 1$ words between the external memory and the internal memory.
The I/O complexity of an algorithm for the external memory model is simply the number of blocks moved between the two memories. We refer to the excellent survey in~\cite{Vitter06} for a more exhaustive explanation. 

The time of some of the previous TCU algorithms recall the I/O complexity of the respective external memory algorithms.
For instance, the cost of dense matrix multiplication with only semiring operations (Theorem~\ref{thm:mmult2}) is $\BO{n^{3/2}/\sqrt{m}}$ when $\ell = \BO{1}$, while the I/O complexity for the same problem in the external memory model is $\BO{n^{3/2}/\sqrt{M}}$ when $B=\BO{1}$~\cite{Vitter06}.

We observe that the product between two matrices of size $\sqrt{m}\times \sqrt{m}$ requires $\BO{m}$ I/Os to load and storing the input in an internal memory with $M=3m$ and $B=\BO{1}$. Therefore any call to the tensor unit in a TCU can be simulated in the external memory of size $M=3m$ with $\BT{m}$ I/Os. 
Leveraging on this claim, we show that a lower bound in the external memory model translates into a lower bound in a weaker version of the TCU model.
In the \emph{weak TCU model},  the tensor unit can only multiply matrices of size $\sqrt{m}\times \sqrt{m}$ (i.e., we cannot exploit tall left matrices).
We observe that any algorithm for the original TCU model can be simulated in the weak version with a constant slowdown when $\ell = \BO{m}$: indeed, the multiplication between an $n\times \sqrt{m}$ matrix with a $\sqrt{m}\times \sqrt{m}$ can be executed in the weak model by splitting the $n\times \sqrt{m}$ matrix into $n/\sqrt{m}$ matrices of size $\sqrt{m}\times \sqrt{m}$ and then performing $n/\sqrt{m}$ matrix multiplications with total time $\BO{n\sqrt{m}}$.

\begin{theorem}
Consider a computational problem $\mathcal P$ with a lower bound $F_\mathcal{P}$ on the  I/O complexity in an external memory with memory size $M=3m+O(1)$ and block length $B=1$.
Then, any algorithm for $\mathcal{P}$ in the weak TCU model requires $\BOM{F_{\mathcal{P}}}$ time.
\end{theorem}
\begin{proof}
Consider an algorithm for the weak $(m,\ell)$-TCU, and let $T=T_t + T_o$ be the total running time with $T = \SO{F_\mathcal{P}}$: we denote with $T_t$ the running time due to tensor units, and with $T_o$ all the remaining operations.
We can simulate the algorithm on an external memory with $M=3m+O(1)$ as follows.
All standard operations are simulated using $O(1)$ internal memory and incurring $\BO{T_o}$ I/Os.
Each call to the tensor unit can be simulated in the external memory by loading the two input matrices in the internal memory with $\BO{m}$ I/Os, computing the product  with no I/Os, and then writing the result in external memory with  $\BO{m}$ I/Os.
If we have $k$ calls to the tensor unit, the algorithm requires $k m = \BO{T_t}$ I/Os (recall that each call requires $\BT{m}$ time).
Thus the TCU algorithm gives an external memory algorithm with I/O complexity $\BO{T_t+T_o} = \SO{F_{\mathcal{P}}}$, which is a contradiction. 
Therefore, we must have $T=\BOM{F_\mathcal{P}}$.
\end{proof}

\section{Conclusion}
In this paper, we have introduced the first computational model for tensor core units, namely the $(m,\ell)$-TCU model.
We have used this model for designing and analyzing several algorithms from linear algebra, broadening the class of algorithms that benefit of a fast hardware circuit for matrix multiplication.
The paper leaves several open questions:
\begin{itemize}
    \item The TCU model should be experimentally validated. Do experimental performance behave as predicted in the theoretical model? Can we use the TCU model for effectively model Google TPUs and NVIDIA TCs?
    \item Which other computational problems can benefit of a tensor unit? Can existing algorithms for deep learning on tensor cores be further improved?
    \item Hardware accelerators have parallel tensors and low numerical precision. How can we include these features in the TCU model, and to what extent do they affect TCU algorithm design? 
\end{itemize}

\section*{Acknowledgments}
This work was partially supported by NSF grant CNS-1553510, UniPD SID18 grant, PRIN17 20174LF3T8 AHeAd,  MIUR Departments of Excellence, UniBZ-CRC 2019-IN2091 Project, and  INdAM-GNCS Project 2020 NoRMA. 
Some results are based upon work performed at the AlgoPARC Workshop on Parallel Algorithms and Data Structures at the University of Hawaii at Manoa, in part supported by the NSF Grant CCF-1930579.

\bibliographystyle{abbrv}
\balance
\bibliography{biblio}

\end{document}